\documentclass[a4paper,article,10pt]{memoir}
\usepackage[english]{babel}
\usepackage{CCLAuthor}
\usepackage[reqno]{amsmath}
\usepackage{amssymb}
\usepackage{amsthm}  
\theoremstyle{plain}
\newtheorem{theorem}{Theorem}

\newtheorem{lemma}[theorem]{Lemma}

\theoremstyle{definition}

\newtheorem{remark}[theorem]{Remark}

\numberwithin{theorem}{chapter}
\usepackage[round,authoryear,comma]{natbib} 
\usepackage[]{graphicx}               
\usepackage{xcolor}
\usepackage{tikz}
\usepackage{pgfplots}
\pgfplotsset{compat=1.18}
\usetikzlibrary{patterns,hobby}

\definecolor{mycolor1}{rgb}{0.00000,0.44700,0.74100}%
\definecolor{mycolor2}{rgb}{0.85000,0.32500,0.09800}%
\definecolor{mycolor3}{rgb}{0.92900,0.69400,0.12500}%

\usepackage{enumitem}

\newtheorem{exercise}[theorem]{Exercise}

\newcommand{\dd}{\mathop{}\!\mathrm{d}}

\newcommand{\F}{\mathcal{F}}
\newcommand{\ee}{\mathrm{e}}
\renewcommand{\Re}{\mathop{}\!\mathrm{Re}}

\begin{document}
%
%
%
\title{Size-Structured Population Dynamics\thanks{This chapter is intended to become part of a volume edited by D. Breda, R. Vermiglio and J. Wu in the Springer Book Series CISM International Centre for Mechanical Sciences, following the school ``Delays and Structures in Dynamical Systems: Modeling, Analysis and Numerical Methods'', November 20--24, 2023.}}
%
%
\author{%
    Odo Diekmann\textsuperscript{a}
    and 
    Francesca Scarabel\textsuperscript{b,c}
    \\ \smallskip\small
    \textsuperscript{a} 
    Department of Mathematics, Utrecht University, Utrecht, The Netherlands                
    \\
    \textsuperscript{b} 
    School of Mathematics, University of Leeds, Leeds, United Kingdom
    \\
    \textsuperscript{c}
    Computational Dynamics Laboratory,
    Department of Mathematics, Computer Science and Physics,
    University of Udine, Italy
    }
    \maketitle
%
%
%
%
    \begin{abstract}
This chapter focuses on variable maturation delay or, more precisely, on the mathematical description of a size-structured population consuming an unstructured resource. When the resource concentration is a known function of time, we can describe the growth and survival of individuals quasi-explicitly, i.e., in terms of solutions of ordinary differential equations (ODE). Reproduction is captured by a (non-autonomous) renewal equation, which can be solved by generation expansion. After these preparatory steps, a contraction mapping argument is needed to construct the solution of the coupled consumer--resource system with prescribed initial conditions. As we shall show, this interpretation-guided constructive approach does in fact yield weak solutions of a familiar partial differential equation (PDE).
A striking difficulty with the PDE approach is that the solution operators are, in general, not differentiable, precluding a linearized stability analysis of steady states. This is a manifestation of the state-dependent delay difficulty. As a (not entirely satisfactory and rather technical) way out, we present a delay equation description in terms of the history of both the $p$-level birth rate of the consumer population and the resource concentration. 
We end by using pseudospectral approximation to derive a system of ODE and demonstrating its use in a numerical bifurcation analysis. Importantly, the state-dependent delay difficulty dissolves in this approximation.
\end{abstract}

\begin{figure}[t]
  \hfill
  \begin{minipage}{0.4\textwidth}
    \includegraphics[height=.25\textheight]{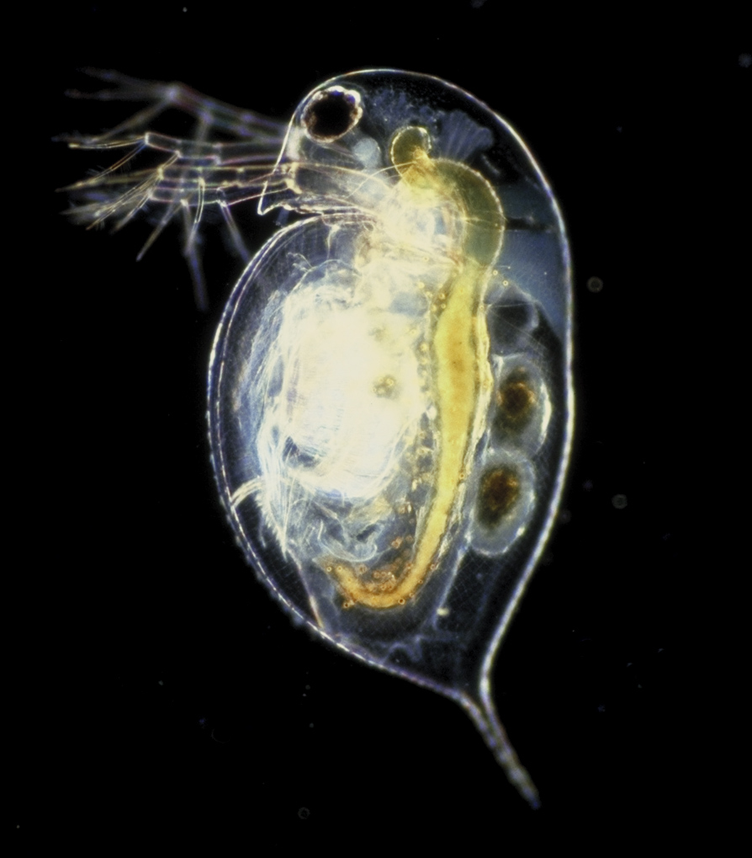}
  \end{minipage}%
  \begin{minipage}{0.5\textwidth}
    \caption{\emph{Daphnia pulex}. \\
    Source: \citet{Daphnia-source}. \\
    Author: Paul Hebert. \\
    Available under the Creative Commons Attribution License. 
    }
  \end{minipage}\hfill
\end{figure}

\CCLsection{Introduction}

Often aquatic toxicity tests are performed with \emph{daphnids} (i.e., water fleas) as the test animal. In such tests, daphnids are cultured while having plenty of food (they filter feed on algae) and while being exposed to a specific concentration of a chemical substance. The question `is it possible to predict the effect of the chemical  substance on populations in the wild, where food is sometimes rather scarce, on the basis of the outcome of the lab tests?' ultimately led to the comprehensive Dynamic Energy Budget (DEB) theory for metabolic organization, see~\cite{Kooijman2010book}.

In the wake of this development, size-structured population models were formulated and analyzed, see for instance \cite{Diekmann2010, Diekmann2017erratum, MetzDiekmann1986}, and the references given there. 

Why size? A peculiar property of daphnids is that they are born with a fixed size $x_b$ and that they mature (i.e., become able to reproduce) upon reaching another fixed size $x_A$ (when assuming that the sizes $x_b$ and $x_A$ are the same for all individuals, we simplify a bit, but not that much; the benefit is that we keep the model parameter scarce). The growth rate $g$ (i.e., the rate at which the size of an individual increases) depends on the size $x$ and the food concentration $S$. If $S$ is constant in time, there is a fixed relationship between the size and the age of an individual and we are in the situation described by \citet{DiekmannScarabelAge}.
In particular, there is a fixed maturation delay corresponding to the time needed to grow from size $x_b$ to size $x_A$.

In general, however, the maturation delay will vary in time, for the simple reason that the food concentration does. Fluctuations of the food concentration may be due to seasonal effects, but may also be brought about by the effects of consumption. Taking a daphnid-centered point of view, we call the food concentration the \emph{environmental variable} and we classify consumption of food as \emph{feedback to the environmental variable}. 

It turns out that size structure and density dependence by way of variable maturation delay are of particular relevance for the ecology of many fish species. The book by \cite{DeRoosPersson2013} provides motivating observations, lots of examples and an inspiring survey of a large body of literature.

The above observations are our motivation to consider in this chapter models for the interaction of a size-structured consumer with an unstructured resource. We do assume that all consumer individuals are born with the same size, indicated by $x_b$, and that the growth rate $g$ is a strictly positive function of size $x$ and food concentration $S$ (so individuals cannot shrink). 
We start by considering $S$ as a given function of time. Our aim is to:
\begin{itemize}
    \item introduce various ways of formulating size-structured models;
    \item show how these formulations relate to each other;
    \item outline relevant mathematical theory, also for an analysis of the feedback loop that arises when the dynamics of $S$ is partly determined by the `grazing' of the consumer population (we shall provide pointers to the literature for precisely stated results and proofs);
    \item discuss steady states and their stability, with due attention for the (at first surprising) fact that linearized stability is a subtle issue for size-structured models (once one `sees' the connection with state-dependent delay, the fact is less surprising for specialists in the theory of \emph{delay differential equations} (DDE));
    \item describe how pseudospectral approximation facilitates numerical bifurcation analysis.
\end{itemize}
Our presentation shall often be sketchy, but we provide references to the literature for more precise formulations and technical details.  

\bigskip
The chapter is organized as follows. In the next section, we introduce notational preliminaries for the $i$-level and $p$-level dynamics. Here and throughout the chapter the prefix $i$ refers to `individual', while $p$ refers to `population'. 
In Section~\ref{sec:interpretation} we assume that the food concentration is prescribed, and use the interpretation to express the population size-density $m(t,\cdot)$ in terms of the population birth rate and the initial density. 
Our interpretation-driven construction is equivalent to the more traditional \emph{partial differential equation} (PDE) approach. This is rigorously shown in Section~\ref{sec:PDE}.
In Section~\ref{s:RE} we focus on the $p$-level birth rate~$b$, and derive a \emph{renewal equation} (RE) that fully determines the population dynamics if the history of~$b$ for $t\leq 0$ is assumed to be known. 
In Section~\ref{sec:feedback} we relax the assumption that the food concentration~$S$ is given, and derive a differential equation for~$S$ that accounts for the population feedback through consumption. 

Sections~\ref{sec:steadystates} and \ref{sec:dynamicalsystem} adopt a dynamical system perspective. First, we show how steady states of the infinite-dimensional system are still determined by one single equation in one unknown for the food concentration, stating that the basic reproduction number should be one. Then, we discuss the difficulties inherent with the study of stability of steady states for the dynamical system defined by the PDE, stemming from the smoothness requirement on the initial density~$m_0$. We summarize the delay equation approach as a way around this technical challenge. 

Section~\ref{sec:pseudospectral} summarizes the pseudospectral approximation as a user-friendly method to perform numerical bifurcation analyses of size-structured models formulated as PDE. The infinite-dimensional system is reduced to a finite-dimensional one, which is amenable to be studied with widely available software for ordinary differential equations (ODE). 

The chapter ends with an Appendix on infinite-delay equations and the pseudospectral approximation for their numerical bifurcation analysis. We highlight the challenges introduced by the infinite delay as well as the limitations of the method for practical applications.

\CCLsection{Conceptual and Notational Preliminaries}
\label{sec:preliminaries}

\CCLsubsection*{Part One: $i$-Level}
\noindent Let $X_S(t,s,\xi)$ denote the size of an individual at time $t$, given that its size equals $\xi$ at time $s$ and given the food concentration $S$ on the interval with end points $s$ and $t$. We assume that $X_S(\cdot,s,\xi)$ is equal to $x$ as a function of $t$, where $x(t)$ is determined by solving the ODE
\begin{equation} \label{x-ODE}
\frac{\dd x}{\dd \tau} = g(x,S(\tau))
\end{equation}
with initial condition
\begin{equation} \label{x-IC}
x(s) = \xi.
\end{equation}

We assume that $g$ takes positive values for the relevant values of $S$ and $x$ (admittedly this statement is a bit vague, but making it more precise would divert the attention from what matters most) and that the size at birth, denoted by $x_b$, is positive. This entails that $x_b$ is the minimal size that can occur and that, more generally, $t \mapsto X_S(t,s,\xi)$ is monotone increasing.

Let $T_S(x,s,\xi)$ denote the time at which the size of an individual equals~$x$, given that its size equals~$\xi$ at time~$s$, and given~$S$ on the interval with end points~$s$ and~$t$. In other words, $x \mapsto T_S(x,s,\xi)$ is the inverse function of $t \mapsto X_S(t,s,\xi)$.

Let $\F_S(t,s,\xi)$ be the probability that an individual having size~$\xi$ at time~$s$ is still alive at time~$t$ (so here, in contrast with before, we require that $t \geq s$). Whenever the model specification involves a death rate $\mu(x,S)$, this survival probability is `explicitly' given by
\begin{equation} \label{F-mu}
    \F_S(t,s,\xi) = \exp \left\{ - \int_s^t \mu(X_S(\tau,s,\xi), S(\tau)) \dd \tau \right\}, \quad t\geq s.
\end{equation}

\CCLsubsection*{Part Two: $p$-Level}
\noindent Let $M(t,x)$ denote the number of individuals with size less than or equal to~$x$ at time~$t$.
Let $m(t,x)$ denote the partial derivative of~$M$ with respect to~$x$. So $m(t,\cdot)$ is the population size-density, i.e., the number of individuals per unit of size at size~$x$ and time~$t$.
Let $b(t)$ denote the $p$-level birth rate at time~$t$. So $b(t)$ is the rate ($=$ number$/$time) at which newborn individuals enter the population at time~$t$. All of these newborns have size~$x_b$.

\CCLsection{Using the Interpretation to Formulate Bookkeeping Consistency Relations}
\label{sec:interpretation}

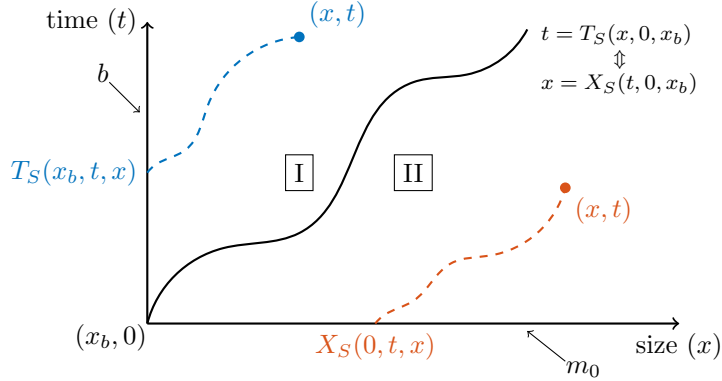
\begin{figure}[t]
    \centering

 
    \begin{tikzpicture}
    \draw [->,thick] (0,0) -- (7,0) node [anchor=north] {size ($x$)};
    \draw [->,thick] (0,0) -- (0,4) node [anchor=east] {time ($t$)};

    \draw [thick,dashed,color=mycolor1] (0,2) node [anchor=east] {$T_S(x_b,t,x)$} to [curve through ={(0.1,2.1)  . . (0.5,2.3) . . (0.8,2.9) }] (2,3.8);
    \fill [color=mycolor1] (2,3.8) circle (2pt) node [anchor=south west,color=mycolor1] {$(x,t)$}; 

    \draw [thick] (0,0) node {} to [ curve through ={(1,1)  . . (2,1.2) . . (3.5,3.2) . . (4.2,3.3)}] (5,3.9);
    \draw (5,3.5) node [anchor=south west, inner sep=5pt] { \footnotesize
        $t = T_S(x,0,x_b)$
    }; 
    \draw (5,3.5) node [rotate=90,yshift=-35pt] {\scriptsize $\Leftrightarrow$ };
    \draw (5,3.5) node [anchor=north west, inner sep=5pt] { \footnotesize
        $x= X_S(t,0,x_b)$
    }; 

    \draw [thick,dashed,color=mycolor2] (3,0) node [anchor=north] {$X_S(0,t,x)$} to [curve through ={(3.1,0.1)  . . (3.5,0.3) . . (4,0.8) . . (4.5,0.9)  }] (5.5,1.8);
    \fill [color=mycolor2] (5.5,1.8) circle (2pt) node [anchor=north west,color=mycolor2] {$(x,t)$}; 

    \draw [->] (-0.5,3.2) node [anchor= south east,inner sep=0pt] {$b$} -- (-0.1,2.8) ;
    \draw [->] (5.5,-0.5) node [anchor=north west,inner sep=0pt] {$m_0$} -- (5,-0.1) ;
    
    \node [anchor=north east,inner sep=0pt] {$(x_b,0)$};

    \draw (2,2) node [draw] {I};
    \draw (3.5,2) node [draw] {II};

    \end{tikzpicture}
    \caption{Size-time quarter plane $(x,t)$ with $t\geq 0$ and $x \geq x_b$, divided into two parts, I and II, by the curve $t=T_S(x,0,x_b)$, or, equivalently, $x=X_S(t,0,x_b)$. 
    The birth rate $b$ is given on $x=x_b$ and the initial size-density $m_0$ is given on $t=0$. 
    }
    \label{fig:size}
\end{figure}

Throughout this section we consider~$S$ as a known function of time. In order to prepare for a mathematical analysis of the initial value problem, we also consider the size-density $m(t,x)$ as a known function of~$x$ for~$t$ equal to the chosen initial time. Without loss of generality we choose $t=0$ as the initial time. So we require that
\begin{equation} \label{m0}
    m(0,x) = m_0(x), \quad \text{for } x \geq x_b,
\end{equation}
with $m_0$ a given/known integrable non-negative function. 
     
For the time being, we pretend that we know the population level birth rate~$b$ for $t \geq 0$. With reference to Figure~\ref{fig:size}, we can now say that we know both the initial distribution along the boundary $t=0$ as well the input along the boundary $x=x_b$. The task is to determine $m(t,x)$ in the interior of the quarter plane from these data. In order to perform this task, we first express $M(t,x)$ explicitly in terms of the data, i.e., $b$ and $m_0$, and the functions~$X_S$, $T_S$ and~$\F_S$. This may seem a detour, but the point is that the formula for~$M$ can be deduced straightforwardly from the interpretation, as we show now.
     
Let $x$ and $t$ be given. If $t > T_S(x,0,x_b)$, the individuals with size less than or equal to $x$ were born after time zero. In fact after the time $T_S(x_b,t,x) > 0$ at which the individuals having size $x$ at time $t$ were born. Taking the possibility of death into account, we arrive at the conclusion that
\begin{equation} \label{MpartI}
    M(t,x) = \int_{T_S(x_b,t,x)}^t \F_S(t,\tau,x_b) b(\tau) \dd\tau, \quad \text{if } t > T_S(x,0,x_b).
\end{equation}

Next assume that $t < T_S(x,0,x_b)$. Then all the individuals born at some time $\tau \in [0,t]$ contribute to $M(t,x)$, provided they survived in between $\tau$ and $t$. The other contribution to $M(t,x)$ comes from individuals that were already present at time zero and did not die before time $t$. An individual with size $x$ at time $t$ had size $X_S(0,t,x)$ at time zero. And in order to be smaller than $x$ at time $t$, an individual should have been smaller than $X_S(0,t,x)$ at time zero. Taking once again the possibility of death into account, we obtain
\begin{multline} \label{MpartII}
    M(t,x) = \int_0^t \F_S(t,\tau,x_b) b(\tau) \dd\tau + \int_{x_b}^{X_S(0,t,x)} \F_S(t,0,\xi) m_0(\xi) \dd\xi, \\ \text{if } t < T_S(x,0,x_b).
\end{multline}

In order to express $m(t,x)$ explicitly in terms of $b$, $m_0$, $X_S$, $T_S$ and $\F_S$, we only need to differentiate with respect to $x$. Thus we obtain
\begin{equation} \label{mtx}
m(t,x) = 
{\footnotesize
\begin{cases}
    \F_S(t,0,X_S(0,t,x)) \, m_0(X_S(0,t,x)) \, D_3X_S(0,t,x), & t < T_S(x,0,x_b), \\[8pt]
    \F_S(t,T_S(x_b,t,x),x_b) \, b(T_S(x_b,t,x)) \, (- D_3T_S(x_b,t,x)), & t > T_S(x,0,x_b).
\end{cases}}
\end{equation}

Here the derivatives of, respectively, $X_S$ and $T_S$ reflect that $m$, $m_0$ and $b$ are not numbers, but numbers per unit of size ($m$ and $m_0$) or per unit of time~($b$). (We invite readers who are familiar with thinking in terms of physical dimensions to check that both expressions at the right-hand side of~\eqref{mtx} have the same dimension as~$m$.) 
Apart from these factors, the two identities express that the individuals having size~$x$ at time~$t$ are the surviving fraction of the individuals that, respectively, 
\begin{enumerate}[label=(\roman*)]
    \item had size $X_S(0,t,x)$ at time zero;
    \item were born at time $T_S(x_b,t,x)$.
\end{enumerate} 
The factors involving the derivatives take care of the transformation of intervals (intervals on the size axis at time zero are mapped to intervals on the size axis at time $t$; intervals on the time axis at $x=x_b$ are mapped to intervals on the size axis at time~$t$).

We conclude that:
\begin{itemize}
    \item measures (here represented by the normalized bounded variation function $M(t,\cdot)$) are in some respects easier to work with than densities;
    \item the biological interpretation yields an explicit formula for $m$ in terms of $m_0$, $b$, $T_S$, $X_S$, and $\F_S$ (but keep in mind that, so far, we eliminated density dependence by considering $S$ as a given function of time). 
\end{itemize}

\CCLsection{How Does This Relate to the PDE Approach?}
\label{sec:PDE}

The traditional way to derive \eqref{mtx} is to solve (by way of integration along characteristics) the first order PDE
\begin{equation} \label{m-PDE}
    \frac{\partial m}{\partial t} + \frac{\partial (gm)}{\partial t} = - \mu m,
\end{equation}
with initial condition \eqref{m0}, here repeated as
\begin{equation} \label{m-IC}
m(0,x) = m_0(x),
\end{equation}
and boundary condition  
\begin{equation} \label{m-BC}
    g m |_{x=x_b} = b. 
\end{equation}

Rather than demonstrating how to derive \eqref{mtx} in this manner, we now show that \eqref{mtx} does indeed  specify a solution, in an appropriate weak sense, of \eqref{m-PDE}--\eqref{m-BC}. We do this as a service to those readers who prefer mathematical manipulation over reasoning in terms of individuals. But we like to emphasize that there is, in fact, no need to formulate the PDE \eqref{m-PDE} and to specify in which sense \eqref{mtx} is its unique solution satisfying \eqref{m-IC} and \eqref{m-BC}, for the simple reason that it is easy to understand \eqref{MpartI} and \eqref{MpartII} and to verify that \eqref{mtx} follows by differentiation. 

A first aspect of the `weak sense' is that we focus on the equation 
\begin{equation} \label{M-PDE}
    \frac{\partial}{\partial t} M(t,x) + g(x,S(t)) \frac{\partial}{\partial x} M(t,x) = b(t) - \int_{x_b}^x \mu(\xi,S(t)) M(t,\dd\xi) 
\end{equation}
obtained by integrating \eqref{m-PDE} with respect to size from $x_b$ to $x$ (we rewrote~$m$ as $\partial M/\partial x$ and wrote the last term as a Stieltjes integral with respect to $M(t,\cdot)$ in order to eliminate the symbol $m$ from the equation and make it self-contained). 
The goal of this section is to prove the following result. 

\begin{theorem}\label{th:4.1}
    The function $M$ defined by \eqref{MpartI} and \eqref{MpartII} satisfies \eqref{M-PDE}. 
\end{theorem}

We first prove a lemma. 
\begin{lemma}\label{l:4.1}
The following hold:
\smallskip
\begin{enumerate}[label=(\roman*)]\addtolength{\itemsep}{0.5\baselineskip}
    \item\label{i} $\left( \frac{\partial}{\partial t}  +  g \frac{\partial}{\partial x}\right) X_S(0,t,x) = 0$;
    \item\label{ii} $\left( \frac{\partial}{\partial t}  +  g \frac{\partial}{\partial x}\right) T_S(x_b,t,x) = 0$.
\end{enumerate}
\end{lemma}

\begin{proof}
The definition of $X_S$ as the solution operator associated with the ODE initial value problem \eqref{x-ODE}--\eqref{x-IC} implies at once that
\begin{equation*}
    X_S(t,r,X_S(r,s,x)) = X_S(t,s,x).
\end{equation*}
Differentiation of this identity with respect to the variable $s$ yields
\begin{equation*}
    D_3X_S(t,r,X_S(r,s,x))\, D_2X_S(r,s,x) = D_2X_S(t,s,x)
\end{equation*}
and, by putting $s=r$,
\begin{equation*}
    D_3X_S(t,r,x) \, D_2X_S(r,r,x) = D_2X_S(t,r,x).
\end{equation*}
By differentiating the identity $X_S(r,r,x) = x$ with respect to $r$ we find that
\begin{equation*}
    D_1X_S(r,r,x) + D_2X_S(r,r,x) = 0,
\end{equation*}
and from \eqref{x-ODE} we deduce that 
\begin{equation*}
    D_1X_S(r,r,x) = g(x,S(r)).
\end{equation*}
Combining the last three identities we obtain
\begin{equation*}
    D_2X_S(t,r,x) + g(x,S(r)) \, D_3X_S(t,r,x) = 0,
\end{equation*}
and if we now put $t=0$ and subsequently replace $r$ by $t$ we obtain \ref{i}.

To prove \ref{ii}, we likewise first observe that
\begin{equation*}
    T_S(x,T_S(z,t,y),z) = T_S(x,t,y),
\end{equation*}
and next differentiate with respect to $z$. This yields
\begin{equation*}
    D_2T_S(x,T_S(z,t,y),z) \, D_1T_S(z,t,y) + D_3T_S(x,T_S(z,t,y),z) = 0
\end{equation*}
and, by putting $z=y$,
\begin{equation}\label{*}
    D_2T_S(x,t,y) \, D_1T_S(y,t,y) + D_3T_S(x,t,y) = 0.
\end{equation}
Differentiating the identity
\begin{equation*}
    X_S(T_S(z,t,y),t,y) = z
\end{equation*}
with respect to $z$, we find
\begin{equation*}
    D_1X_S(T_S(z,t,y),t,y) \, D_1T_S(z,t,y) = 1
\end{equation*}
and by putting $z=y$
\begin{equation*}
    D_1X_S(t,t,y) \, D_1T_S(y,t,y) = 1.
\end{equation*}

By \eqref{x-ODE}, the first factor in this last identity equals $g(y,S(t))$, and consequently
\begin{equation*}
    D_1T_S(y,t,y) = \frac{1}{g(y,S(t))}.
\end{equation*}
So, if we multiply \eqref{*} by $g(y,S(t))$, we obtain \ref{ii} modulo a renaming of variables.  
\end{proof}

\begin{proof}[Proof of Theorem \ref{th:4.1}]
    To highlight the key points, we proceed in two steps. 

First we assume that $\mu$ is identically equal to zero or, equivalently, that~$\F$ is identically equal to one. In that special case, \eqref{MpartI} simplifies to 
\begin{equation} \label{MpartI-special}
    M(t,x) = \int_{T_S(x_b,t,x)}^t b(\tau) \dd\tau,
\end{equation}
and \eqref{MpartII} to
\begin{equation} \label{MpartII-special}
    M(t,x) = \int_0^t b(\tau) \dd\tau + \int_{x_b}^{X_S(0,t,x)} m_0(\xi) \dd\xi,
\end{equation}
and equation \eqref{M-PDE} amounts to
\begin{equation} \label{M-PDE-special}
    \left( \frac{\partial}{\partial t}  +  g \frac{\partial}{\partial x}\right) M = b.
\end{equation}

When $\partial/\partial t$ is applied to~$M$ as defined by \eqref{MpartI-special} or \eqref{MpartII-special}, the occurrence of~$t$ as the upper integration boundary yields the right-hand side~$b$ of~\eqref{M-PDE-special}. So~\eqref{M-PDE-special} does indeed hold if both $T_S(x_b,t,x)$ and $X_S(0,t,x)$ are in the kernel of the operator $(\partial/\partial t + g \partial/\partial x)$,
which is the content of Lemma~\ref{l:4.1}. 

We conclude that, in the special case of no mortality, $M$ defined by \eqref{MpartI}--\eqref{MpartII} does indeed satisfy~\eqref{M-PDE}. 
     
How about the general case? The observations concerning the dependence on~$t$ via the upper integration boundary and the identities of Lemma~\ref{l:4.1} are still very relevant. They imply that all we need to do is to consider the differentiation with respect to~$t$ as the first argument of~$\F$ and to verify that this yields the second term at the right-hand side of~\eqref{M-PDE}.

From the interpretation-inspired formula~\eqref{F-mu} (or, equivalently, from the interpretation of $\mu$ as the per capita death rate as a function of $i$-state, i.e.,~ size, and environmental condition, i.e., food concentration) we deduce that
\begin{equation*}
    \frac{\dd}{\dd t} \F_S(t,s,\xi) = - \mu(X_S(t,s,\xi),S(t)) \, \F(t,s,\xi).
\end{equation*}

So if we differentiate \eqref{MpartI} with respect to time, the relevant term reads
\begin{equation*}
    - \int_{T_S(x_b,t,x)}^t \mu(X_S(t,\tau,x_b),S(t)) \F(t,\tau,x_b) b(\tau) \dd\tau,
\end{equation*}
while the corresponding contribution to \eqref{M-PDE} reads
\begin{equation*}
    - \int_{x_b}^x \mu(\xi,S(t)) \F(t,T_S(x_b,t,\xi),x_b) b(T_S(x_b,t,\xi)) D_3T_S(x_b,t,\xi)\dd\xi.
\end{equation*}
The transformation $\tau=T_S(x_b,t,\xi)$, with inverse $\xi=X_S(t,\tau,x_b)$, shows that these two expressions are equal to each other.

If we differentiate \eqref{MpartII} with respect to time, the relevant term reads
\begin{equation*}
    \int_{x_b}^{X_S(0,t,x)} \mu(X_S(t,0,\xi),S(t)) \F_S(t,0,\xi) m_0(\xi) \dd\xi,
\end{equation*}
while the relevant term at the right-hand side of \eqref{M-PDE} reads
\begin{equation*}
    - \int_{X_S(t,0,x_b)}^x \mu(\eta,S(t)) \F_S(t,0,X_S(0,t,\eta)) m_0(X_S(0,t,\eta) D_3 X_S(0,t,\eta) \dd\eta.
\end{equation*}
The transformation $\xi = X_S(0,t,\eta)$ with inverse $\eta = X_S(t,0,\xi)$ shows that these two expressions are equal to each other.

We conclude that \eqref{MpartI}--\eqref{MpartII} does indeed satisfy the PDE \eqref{M-PDE}. Since $M(t,x_b)$ is identically equal to zero, we see, by taking $x=x_b$ in \eqref{M-PDE}, that \eqref{m-BC} holds when we define $m$ as the derivative of $M$ with respect to $x$, cf.~\eqref{mtx}. Finally, note that \eqref{m-IC} is obtained if we choose $t=0$ in \eqref{mtx}, since $X_S(0,0,x) = x$. Thus we verified that $m$ defined by \eqref{mtx} satisfies \eqref{m-PDE}--\eqref{m-BC}, provided we interpret \eqref{m-PDE} in the loose sense that the primitive $M$ of $m$ satisfies the integrated version \eqref{M-PDE}.

\end{proof}

\begin{remark}\label{r:4.2}
The usual way to solve \eqref{m-PDE} by integration along characteristics is to write $\partial (gm)/\partial x$ as
$g \frac{\partial m}{\partial x} + (D_1 g) m$ and to bring $(D_1 g) m$ to the other side of the  equality, thus effectively adding $D_1g$ to $\mu$. This leads to a variant of \eqref{mtx} with a last factor that is, at first sight, rather different. We invite worried readers to check that the first sight is deceptive. For encouragement we state the following lemma. 
\end{remark}

\begin{lemma}\label{l:4.3}
The following hold:
\smallskip
{\small
\begin{enumerate}[label=(\roman*)]
    \item\label{2i} $D_3X_S(0,t,x) = \exp\left\{- \int_0^t D_1g(X_S(\sigma,t,x),S(\sigma)) \dd\sigma\right\}$,
    \item\label{2ii} $- D_3T_S(x_b,t,x) = \frac{1}{g(x_b,t,x)} \, \exp\left\{- \int_{T_S(x_b,t,x)}^t D_1g(X_S(\sigma,t,x),S(\sigma)) \dd\sigma\right\}$.
\end{enumerate}
}
\smallskip\noindent
Hint: To prove \ref{2i}, differentiate \eqref{x-ODE}--\eqref{x-IC} with respect to $\xi$ and integrate.
To prove \ref{2ii}, differentiate the identity $X_S(T_S(x,s,\xi),s,\xi) = x$ with respect to $x$ and use the variant $D_3X_S(t,s,\xi) = \exp \left\{ \int_s^t D_1g(X_S(\sigma,s,\xi),S(\sigma) \dd\sigma\right\}$ of~\ref{2i}.
\end{lemma}

\CCLsection{The Renewal Equation for the $p$-Level Birth Rate}
\label{s:RE}

In this section, we still consider $S$ as a given/known function of time. But now we want to define constructively the function $t \mapsto b(t)$ for $t \geq 0$ on the basis of the initial size-density $m_0$ and a new model ingredient, the per capita rate $\beta$ of giving birth, given the size of the mother and the prevailing food concentration. So the model specification should tell how $\beta$ depends on $x$ and $S$. (In case of an energy budget model, one first specifies the food uptake as a function of $x$ and $S$, and next how the ingested energy is partitioned between maintenance, growth and reproduction. So how exactly $\beta$ depends on $x$ and $S$ is derived from a submodel for ingestion and partitioning of energy.)

The equality
\begin{equation} \label{b-m}
    b(t) = \int_{x_b}^{\infty} \beta(\xi,S(t)) m(t,\xi) \dd\xi
\end{equation}
expresses that the $p$-level birth rate $b$ is simply the addition of all the per capita contributions. Upon substitution of \eqref{mtx} and the obvious transformations of the integration variable, we obtain from \eqref{b-m} the linear~RE
\begin{equation} \label{b-RE}
    b(t) = \int_0^t K(t,\tau) b(\tau) \dd\tau  +  F(t),
\end{equation}
where
\begin{equation*}
    K(t,\tau) := \beta(X_S(t,\tau,x_b),S(t)) \F_S(t,\tau,x_b)
\end{equation*}
and
\begin{equation*}
    F(t) := \int_{x_b}^\infty \beta(X_S(t,0,\xi),S(t)) \F_S(t,0,\xi) m_0(\xi) \dd\xi.
\end{equation*}

\begin{exercise}
Verify that \eqref{b-RE} expresses that the $p$-level birth rate $b$ is composed of contributions by individuals that were born at some time $\tau$ after time $0$ and survived up to the present time $t$, and contributions of individuals that were alive at time $0$, had size $\xi$ at that time, and survived till the present  time $t$. In other words, show that one does not need \eqref{mtx} to be able to formulate~\eqref{b-RE}.
\end{exercise}

\begin{exercise}
In analogy with \citet[Section 3]{DiekmannScarabelAge}, use generation expansion to derive from \eqref{b-RE} a representation of $b$ as an infinite series of well-defined  terms. Readers interested in the two-variable resolvent of the two-variable (in other words, non-autonomous) kernel $K$ are advised to consult the Volterra `bible' by \citet{Gripenberg1990}.
\end{exercise}

We conclude that one can, given the model ingredients $X_S$, $\F_S$ and $\beta$ and given the initial size-density~$m_0$, constructively define the $p$-level birth rate~$b$ for $t \geq 0$ by solving the linear RE~\eqref{b-RE} by means of generation expansion or, in other words, successive approximation. Once~$b$ is constructed, \eqref{mtx} provides an explicit formula for the size-density $m(t,\cdot)$ for $t > 0$.

\CCLsection{Closing the Feedback Loop: a Constructive Definition of $S$ as a Function of Time}
\label{sec:feedback}

So far we developed a systematic methodology: introduce the $i$-state, here~$x$, and the environmental condition, here~$S$; pretend that the environmental condition is a known function of time; use solutions of ODE to describe $i$-state development and survival probability under given environmental conditions; formulate the linear RE for the $p$-level birth rate~$b$ and solve it by generation expansion; and finally derive an explicit formula for (the density of) the $i$-state distribution in terms of the initial distribution and~$b$. 

Now we should address the issue of feedback: how is the environmental condition (partly) determined by interaction with the focal population?  For an attempt at building a general framework see \citet{Diekmann2001formulation},  
and for particular examples see \citet{Calsina1995model, Barril2022}, and \citet{Clement2024well}.
The general idea is to formulate a fixed point problem and to show that the contraction mapping principle can be applied, to obtain a unique solution (probably first on a small time interval, but by continuation on a maximal time interval). The art is to choose both the function spaces and the regularity assumptions concerning the model ingredients such that this approach works. 
(If the maximal time interval is finite, some kind of blow up has to happen; so if biological considerations exclude blow up, one should be able to formulate reasonable assumptions on the model ingredients that guarantee existence and uniqueness of the environmental condition in the time window $[0,\infty)$.)

For the \emph{Daphnia} model, feedback occurs by way of food consumption, and one requires that $S$ satisfies the ODE
\begin{equation} \label{S-m-ODE}
    \frac{\dd S}{\dd t}(t) =  f(S(t)) - \int_{x_b}^\infty \gamma(\xi,S(t)) \, m(t,\xi) \dd\xi
\end{equation}
with initial condition
\begin{equation} \label{S-ODE-IC}
    S(0) = S_0.
\end{equation}
Here, $f$ governs the food dynamics in the absence of consumers (e.g., $f(S) = D(S_{\text{in}} - S)$ in case of chemostat dynamics) and the new model ingredient~$\gamma$ specifies how the per capita uptake rate depends on the intrinsic variable `size' and the extrinsic variable `food concentration'. We refer to \citet{Barril2022} for a precise elaboration of the approach sketched above for this Daphnia model, in particular for the choice of function spaces and for precisely stated assumptions concerning the model ingredients $g$, $\mu$, $\beta$, $\gamma$ and~$f$.


\CCLsection{Steady States}\label{sec:steadystates}
Despite the fact that we deal with an infinite-dimensional dynamical system, finding steady states amounts to solving one equation in one unknown. The reason is that the dynamics of the consumer population is linear when the food concentration is prescribed. So to obtain a steady state, the constant food concentration $\overline{S}$ should be such that the consumer population neither grows nor declines. Recalling from \citet{DiekmannScarabelAge} the notion of \emph{basic reproduction number}, denoted by $R_0$ and describing the expected total number of offspring in an individual's lifetime, we may write this condition in the form
\begin{equation} \label{R0-threshold}
    R_0(\overline{S}) = 1.
\end{equation}

If `the more food, the better' holds, $R_0$ is a strictly monotone increasing function of $\overline{S}$, and \eqref{R0-threshold} has at most one biologically relevant solution. For chemostat dynamics, i.e., $f(S) = D(S_{\text{in}} -S)$, such a solution will exist if $R_ 0(S_{\text{in}}) > 1$, while the consumer is bound to go extinct if $R_0(S_{\text{in}}) < 1$.
In the following, we characterize the function $R_0(\overline{S})$ in terms of the given model parameters, and use \eqref{R0-threshold} to determine the steady state.

\bigskip
Under constant food concentration $\overline{S}$ we have, essentially, an age-struc\-tured problem (since all individuals have the same size $x_b$ at birth): the ODE \eqref{x-ODE} describing the age-size relation is autonomous, hence individuals of age $a$ will have the same size $\bar{x}(a)$, given by the solution of 
\begin{equation*}
    \begin{cases}
        \displaystyle{\frac{\dd \bar{x}}{\dd \tau} = g(\bar{x},\overline{S}),} \\[8pt]
        \bar{x}(0)=x_b, 
    \end{cases}
\end{equation*}
no matter when they were born. 
Let $\overline{n}(a)$ denote the population age-density, and $\overline{m}(x)$ the population size-density, at steady state. These densities are related via 
\begin{equation} \label{nbar-mbar}
    \int_0^a \overline{n}(\alpha)\dd\alpha = \int_{x_b}^{\bar{x}(a)} \overline{m}(\xi)\dd\xi \quad \Rightarrow \quad \overline{n}(a) = \overline{m}(\bar{x}(a)) g(\bar{x}(a),\overline{S}).
\end{equation}
On the other hand, given a constant population birth rate~$\overline{b}$ and the per capita death rate $\mu(\bar{x}(a),\overline{S})$, the age-density is given by  
\begin{equation} \label{nbar-bbar}
    \overline{n}(a) = \overline{b} \, \ee^{-\int_0^a \mu(\bar{x}(\alpha),\overline{S}) \dd\alpha} . 
\end{equation}
Combining \eqref{nbar-mbar} and \eqref{nbar-bbar}, we conclude that $\overline{m}$ and $\overline{b}$ are related via
\begin{equation} \label{mbar-bbar}
    \overline{m}(\bar{x}(a)) = \frac{\overline{b}}{g(\bar{x}(a),\overline{S})} \,\ee^{-\int_0^a \mu(\bar{x}(\alpha),\overline{S}) \dd\alpha}. 
\end{equation}
Finally, from the interpretation of $R_0$ (cf.~equation (24) in \citet{DiekmannScarabelAge}), 
we can characterize $R_0(\overline{S})$ as 
\begin{align}
    R_0(\overline{S}) &= \int_0^\infty \beta(\bar{x}(a),\overline{S}) \, \ee^{-\int_0^a \mu(\bar{x}(\alpha),\overline{S}) \dd\alpha} \dd a \notag \\
        &= \int_{x_b}^\infty \frac{\beta(y,\overline{S})}{g(y,\overline{S})} \, \ee^{-\int_{x_b}^y \frac{\mu(\eta,\overline{S})}{g(\eta,\overline{S})} \dd\eta} \dd y. \label{R0}
\end{align}
The value of~$\overline{S}$ is then determined by the condition~\eqref{R0-threshold}. 

To keep the food concentration at~$\overline{S}$, the production of food per unit of time, as described by $f(\overline{S})$, needs to be balanced exactly by the consumption of food per unit of time. Using the size-age relation \eqref{mbar-bbar} in equation \eqref{S-m-ODE}, we find
\begin{equation*}
    f(\overline{S}) = \int_{x_b}^\infty \gamma(\xi, \overline{S}) \,\overline{m}(\xi) \dd\xi = \overline{b} \int_{x_b}^\infty \frac{\gamma(\xi,\overline{S})}{g(\xi,\overline{S})} \,\ee^{-\int_{x_b}^\xi \frac{\mu(\eta,\overline{S})}{g(\eta,\overline{S})} \dd\eta} \dd \xi,
\end{equation*}
with the food concentration at the constant level $\overline{S}$ and the consumer birth rate equal to a, as yet unknown, constant $\overline{b}$.
Linearity in $\overline{b}$ makes that we can use this condition to derive the explicit expression
\begin{equation*}
    \overline{b} = \frac{f(\overline{S})}{\int_{x_b}^\infty \frac{\gamma(y,\overline{S})}{g(y,\overline{S})} \,\ee^{ - \int_{x_b}^y \frac{\mu(\eta,\overline{S})}{g(\eta,\overline{S})} \dd\eta } \dd y} . 
\end{equation*}

We conclude that, compared to unstructured consumer--resource models, finding/characterizing steady states is hardly more difficult when the consumer population is structured by size (and size at birth is fixed). As we shall see shortly, this is anything but true when it comes to stability.





\CCLsection{The Dynamical Systems Perspective} \label{sec:dynamicalsystem}
In Section \ref{sec:feedback} we described how $m(t,\cdot)$ and $S(t)$ can be constructed from the model ingredients and the initial data $m_0$ and $S_0$. The map $(m_0,S_0) \mapsto (m(t,\cdot),S(t))$ defines a dynamical system on a state space that is the Cartesian product of (the positive cone in) a space of functions of the size variable~$x$ and $\mathbb{R}_+$.
In the preceding section we found that it is rather easy to characterize the steady states of this nonlinear dynamical system. What can we say about the \emph{stability} of these steady states?

The designated way to establish the stability character of a steady state proceeds in three or four steps:
\begin{enumerate}
    \item linearize the equations around the steady state;
    \item for the linearized equations, characterize (in)stability in terms of spectral properties, in particular eigenvalues;
    \item if possible, show that the eigenvalues coincide exactly with the roots of a characteristic equation (for delay equations this is possible);
    \item analyze the position in the complex plane of the roots of the characteristic equation (and, in a bifurcation analysis, how positions change when parameters are varied).
\end{enumerate}

When the first step is taken formally, one implicitly presupposes that the solution of the nonlinear problem depends on the initial condition in a differentiable manner. Now recall the upper expression in \eqref{mtx}. The argument of $m_0$ depends on the values of $S$ in the interval $[0,t]$, so in particular on~$S_0$. No matter how smoothly this argument depends on~$S_0$, some regularity of~$m_0$ is needed in order to obtain differentiability of~$m$ with respect to~$S_0$ (for instance, if we view $m(t,\cdot)$ as an element of~$L^1$, we would need absolute continuity of~$m_0$; see \citet{Barril2022}).
This is a manifestation of the characteristic difficulty of state-dependent delay problems: when a function is translated with variable speed, we need some smoothness of that function when we want to differentiate with respect to a variable that influences the speed. (For an ingenious, but technically demanding, way to overcome this difficulty see \citet{Hartung2006functional}.) 

Note, incidentally, that when we consider $S$ as a function of time satisfying equation \eqref{S-m-ODE}, the lack of differentiability with respect to the initial condition is reflected in the fact that the PDE~\eqref{m-PDE} is quasi-linear, meaning that the term with the highest derivative involves nonlinearity.

It seems we encounter a major stumbling block. Can we get around it? An affirmative answer can be found by resorting to a `delay' dynamical systems perspective as follows. Rather than prescribing the initial size density~$m_0$, we prescribe the history of the $p$-level birth rate~$b$. This forces us to also describe the history of the resource concentration~$S$, since we need to be able to calculate the current (i.e., at time zero) size of the individuals that were born some time ago. So we need to specify 
\begin{equation} \label{IC-bS}
\begin{cases}
    b(\theta) = \phi(\theta) \\
    S(\theta) = \psi(\theta)
\end{cases}
\theta \leq 0,
\end{equation}
in order to get started. The governing equations
\begin{equation} \label{DE-bS}
    \begin{cases}
        \displaystyle{b(t) = \int_{-\infty}^t \beta(X_S(t,\tau,x_b),S(t)) \F_S(t,\tau,x_b) b(\tau) \dd \tau} \\[5pt]
        \displaystyle{\frac{\dd}{\dd t}S(t) = f(S(t)) - \int_{-\infty}^t \gamma(X_S(t,\tau,x_b),S(t)) \F_S(t,\tau,x_b) b(\tau) \dd \tau} \\
    \end{cases}
\end{equation}
are a coupled system of a RE for~$b$ and a DDE for~$S$. Given our earlier work, there is no need to discuss the construction of solutions of \eqref{IC-bS}--\eqref{DE-bS}, since we can, given $\phi$ and $\psi$, define $m_0$ by the lower expression in \eqref{mtx}, with $t=0$, $\phi$ substituted for $b$ and $\psi$ substituted for $S$, and define $S_0$ by $\psi(0)$ in order to find back the initial conditions that we considered in the earlier sections.

The map $(\phi,\psi) \mapsto (b_t,S_t)$  defines a dynamical system on the Cartesian product of suitably defined (infinite delay) history spaces for, respectively,~$b$ and $S$. We refer to \cite{Barril2022} for a rather technical proof that, given somewhat restrictive assumptions concerning the model ingredients, the right-hand side of \eqref{DE-bS} can be written as a $C^1$ map (from the state space to $\mathbb{R}^2$) applied to $(b_t,S_t)$. As a consequence, the solution operators are differentiable in this setting and the validity of the \emph{principle of linearized stability} follows from~\cite{Diekmann2012Blending}. 

Of course it is rather far-fetched to assume that $\phi$ and $\psi$ in \eqref{IC-bS} are known functions at the start of an experiment. But if one is monitoring the dynamics of the consumer population and its resource already for some time, much of the history is known. So in the course of time the reasonableness of the bookkeeping scheme in terms of delay equations increases. In particular, qualitative assertions about the asymptotic large time behavior do provide relevant information.

As we noticed above, it is simple to map $(\phi,\psi)$ to $(m_0,S_0)$. But this map is many-to-one and consequently not invertible. \cite{Barril2022} define a pseudo-inverse with $\psi$ identically equal to $S_0$ and use it to transfer (in)stability assertions from the delay setting to the PDE setting. Thus the principle of linearized stability for the PDE formulation is verified in a roundabout way.

The formula \eqref{mtx} consists of two parts. The non-differentiability derives from the upper part. Under reasonable conditions on the death rate, this part decays exponentially in time, no matter how the population as a whole develops. This observation suggests an alternative approach for developing stability and bifurcation theory: consider a setting where the solution operators are the sum of two operators, one obeying uniform exponential estimates and the other depending on the initial condition in a differentiable manner. Perhaps this is a promising approach, perhaps it is a cul-de-sac.

The delay equation formulation \eqref{DE-bS} involves at its right-hand side a map that sends $(b_t,S_t)$ to $\mathbb{R}^2$. Autonomous delay equations go hand in hand with characteristic equations exactly because the rule for extension has finite-dimensional range. For the \emph{Daphnia} model, the characteristic equation is derived and analyzed by \citet{Diekmann2010, Diekmann2017erratum}.
In Sections 5 and 6 of~\cite{Diekmann2010}, various (modest) ecological insights are derived by delineating the stability boundary of the unique steady state in a two-dimensional parameter space. At the Hopf part of the stability boundary one can, with a bit of effort, determine whether there exists a stable periodic solution outside of (but near to) the stability region or, alternatively, an unstable periodic solution inside (but near). 
To follow such periodic solutions for parameter values that move away from the stability boundary, one needs numerical continuation and bifurcation tools.
Therefore we now discuss pseudospectral approximation for size-structured models.


\CCLsection{Pseudospectral Approximation for Numerical Bifurcation Analysis}
\label{sec:pseudospectral}

A convenient approach for the numerical bifurcation analysis of size-struc\-tured models is to reduce them to a system of ODE and use the library of tools widely available for ODE to perform the analysis. 
Among the many methods available to perform this reduction, we here focus on pseudospectral approximation. 

To obtain an approximating ODE, one could start from the formulation~\eqref{DE-bS} in combination with pseudospectral approximation of delay equations, as described by \citet{Breda2016Prospects, BMVbook} and \citet{DiekmannScarabelAge}. 
But, when the life cycle involves different stages, the parameters may have discontinuities. It depends on the past environmental condition, when exactly these occur. 
In the \emph{Daphnia} model \citep{Breda2015computing, Diekmann2010}, individuals are born with a fixed size $x_b$ and become mature (hence, fertility becomes strictly positive) when they reach a given size $x_A$. The maturation age $a_A=a_A(t)$ (the age of individuals becoming mature at time $t$) is defined via the threshold condition
\begin{equation*}
    X_S(a_A,t-a_A,x_b) = x_A,
\end{equation*}
or, equivalently, by $a_A = t-T_S(x_b,t,x_A)$. 
When performing numerical bifurcation analyses, such an implicit condition has to be solved at every continuation step to determine $a_A$, leading to large computation times. 
We refer to \citet{Ando2020} for an ad hoc discretization approach that overcomes this problem by including the variables defining the threshold conditions among the continuation variables, giving an efficient numerical scheme. 

Another challenge of the delay formulation \eqref{DE-bS} is that a priori upper bounds for the support of the survival probabilities might not exist, for instance when the death rate~$\mu$ depends on the size of individuals, which is itself determined via environmental feedback. 
In this case, to apply the approximation described in \citep{DiekmannScarabelAge} one should either truncate the delay (i.e., the support of~$\F$) or resort to a variation of the pseudospectral approximation that accounts for unbounded delay~\citep{Gyllenberg2018, Scarabel2024Infinite}. 
We will discuss this variation in Appendix~\ref{sec:infinite}.

An alternative approach that avoids the state-dependent discontinuities and improves computational efficiency for size-structured models was suggested by \citet{Scarabel2021Vietnam}. 
The idea is to apply the collocation method to the PDE \eqref{M-PDE}, by approximating the state $M(t,\cdot)$, defined on a bounded size interval  $[x_b,x_{\text{max}}]$, with a polynomial of a given degree~$N\in\mathbb{N}_+$. 
In this way, the discontinuities of the model parameters occur at fixed sizes (e.g.,~$x_A$ in the \emph{Daphnia} example) that are assumed to be known.
We first summarize the method, then illustrate it by applying it to a simplified \emph{Daphnia} model. 

\bigskip
Consider a set of $N$ nodes $\{x_1,\dots,x_N\} \subset (x_b,x_{\text{max}}]$, such that
\begin{equation*}
    x_b < x_1 < \cdots < x_N \leq x_{\text{max}},
\end{equation*}
and $N$ corresponding variables $y_j(t)$, $j=1,\dots,N$, for $t\geq 0$. Let $p$ be the $N$-degree polynomial that takes value zero in $x=x_b$ and interpolates the values $y_j$ at the nodes. Using the Lagrange representation, we can write
\begin{equation}\label{pN}
    p(x) = \sum_{j=1}^N y_j \ell_j(x), 
\end{equation}
where the Lagrange polynomials $\ell_j$ are defined by 
\begin{equation*}
    \ell_j(x) = \prod_{\substack{i=0\\ i\neq j}}^N \frac{x-x_i}{x_j-x_i}, \quad j=0,\dots,N,
\end{equation*}
with $x_0=x_b$ by definition. Both~$p$ and $M(t,\cdot)$ take value~$0$ in $x=x_b$. Note also that~$p$ depends on time~$t$ via the coefficients~$y_j$, even though this is not expressed in the notation. 

The pseudospectral approach can be rigorously defined by means of a restriction operator that maps a continuous function to the vector of its values in the collocation points, and a prolongation operator that maps a vector to the interpolating polynomial. In practice, this is equivalent to requiring that the polynomial $p$ satisfies \eqref{M-PDE} for $x \in \{x_1,\dots,x_N\}$. By linearity, we have
\begin{equation*}
    \frac{\partial p}{\partial t} = \sum_{j=1}^N \frac{\dd y_j}{\dd t} \ell_j, \quad \text{and} \quad \frac{\partial p}{\partial x} = \sum_{j=1}^N y_j \frac{ \dd\ell_j}{\dd x}.
\end{equation*}
Hence, using the property $\ell_{j}(x_i)=\delta_{ij}$, where $\delta_{ij}$ is Kronecker's delta, and defining $d_{ij}:= \ell_j'(x_i)$, one obtains from \eqref{M-PDE} the $N$ ODEs
\begin{equation}\label{ODE-yj}
    \frac{\dd y_i}{\dd t} + g(x_i,S) \sum_{j=1}^N d_{ij}y_j = \tilde{b} - \sum_{j=1}^N y_j \int_{x_b}^{x_i} \mu(\xi,S) \ell_j'(\xi) \dd \xi, \quad i=1,\dots,N,
\end{equation}
where $\tilde{b}$ is obtained by replacing in \eqref{b-m} $m$~by the derivative of the approximation~$p$ of~$M$, so (using~\eqref{pN})
\begin{equation*}
\tilde{b} = \sum_{j=1}^N y_j \int_{x_b}^{x_{\text{max}}} \beta(\xi,S) \ell_j'(\xi) \dd \xi.  
\end{equation*}

One can write system~\eqref{ODE-yj} in matrix notation by defining the column vector $y=(y_1,\cdots,y_N)^T$ and matrices $D$, $G(S)$ and $\Sigma(S)$ in $\mathbb{R}^{N\times N}$ as follows 
\begin{align*}
    & D_{ij} = d_{ij} \\
    & (G(S))_{ij} = \begin{cases} g(x_i,S) & \text{if } i=j\\ 0 & \text{otherwise,}\end{cases} \\
    & (\Sigma(S))_{ij}=\int_{x_b}^{x_i} \mu(\xi,S) \ell_j'(\xi) \dd \xi. 
\end{align*}
Then, system \eqref{ODE-yj} can be rewritten as 
\begin{equation}\label{ODE}
    \frac{\dd y}{\dd t} = - G(S) D y - \Sigma(S) y + \tilde{b} \,\mathbf{1},
\end{equation}
where $\mathbf{1}$ is the $N$-dimensional vector with all entries equal to one. 

To close the feedback loop with the environmental variable $S$, one can replace in \eqref{S-m-ODE} the function $m$ by the derivative of $p$, leading to 
\begin{equation}\label{tildeS}
    \frac{\dd \tilde{S}}{\dd t} = f(\tilde{S}) - \sum_{i=1}^N y_i \int_{x_b}^{x_{\text{max}}} \gamma(\xi,\tilde{S}) \ell_i'(\xi)\dd \xi, 
\end{equation}
where we used the notation $\tilde{S}$ to stress that the solution to \eqref{tildeS} is an approximation of $S$. 
If we now replace $S$ in \eqref{ODE} by $\tilde{S}$, system \eqref{ODE}--\eqref{tildeS} can be used to approximate the solutions $(M(t,\cdot),S(t))$ of the full nonlinear problem \eqref{M-PDE} and \eqref{S-m-ODE}. 

A natural choice of collocation nodes to guarantee good convergence of the interpolation scheme is a family of Chebyshev nodes (zeros or extrema of the family of Chebyshev orthogonal polynomials) scaled to the appropriate interval, see for instance \citet[Section 1.5]{Gautschi2004book}.

When parameters have discontinuities (e.g., at the maturation size $x_A$), a piecewise collocation approach can be more accurate, as it also improves the accuracy of quadrature formulas used to numerically compute the integrals. We refer to \citet[Section 2.2]{Scarabel2021Vietnam} for further details.

\bigskip
Differently from the case of RE, whose steady states are constant functions and are interpolated exactly by polynomials of any degree, a steady state $\overline{M}$ of \eqref{M-PDE} can be a function defined on $[x_b,x_{\text{max}}]$, and convergence of the polynomial approximation must be proved. In particular, as shown by \citet[Theorem 1]{Scarabel2021Vietnam}, the order of convergence with respect to the degree $N$ of the collocation polynomial depends on the smoothness of $\overline{M}(x)$ and, thus, on the smoothness of the model parameters. 

To prove theoretically that the stability of the steady states is well approximated by the discrete problems, one should look at the linearized problems, and show that the spectra of the approximating linearized operators converge to the spectrum of the linearized PDE. This is still an open problem. 
It is reasonable to expect that the approximation errors of the eigenvalues exhibit the same order of convergence (with respect to $N$) as the approximation error of the steady state.  
In the absence of a theoretical proof, preliminary numerical investigations, obtained by comparing the results from the approximation of the PDE and the delay equation formulation, have shown that the numerical bifurcation analysis of \eqref{ODE}--\eqref{tildeS} can indeed reveal the properties of \eqref{M-PDE} and \eqref{S-m-ODE}. 
We now report on the results of numerical tests performed with the package MatCont for MATLAB \citep{MatCont2008}. 

\bigskip
Consider the model for \emph{Daphnia} described by  \eqref{M-PDE} and \eqref{S-m-ODE}, with a logistic consumer-free dynamics 
\begin{equation}
    f(S) = a_1 S \left(1-\frac{S}{K}\right),
\end{equation}
and individual parameters as in Table \ref{t:functions_Daphnia}, taken from~\cite{Kooijman1984}. 
The system always admits a consumer-free steady state $(0,K)$, and a nontrivial (positive) steady state $(\overline{b},\overline{S})$ when $R_0>1$, with $R_0$ defined by~\eqref{R0}, with $\overline{S}=K$. 

To perform the bifurcation analysis, we construct an ODE approximation through the pseudospectral approach described above.
Individuals become mature when they reach size $x_A$, after which they start reproducing. To account for the discontinuity in the fertility rate~$\beta$ we used a piecewise approach, with a collocation polynomial of degree $N=10$ in each interval $[x_b,x_A]$ and $[x_A,x_{\text{max}}]$. The resulting system has $20$ equations approximating $M(t,x)$ and the additional equation~\eqref{tildeS}. 
Using MatCont, we here analyze the impact of the maturation size $x_A$, the maximal size $x_{\text{max}}$, and the resource carrying capacity $K$, on the dynamics of the system. The other parameters are defined in Table~\ref{t:parameters}.

\begin{table}[ht]
\renewcommand*{\arraystretch}{1.5}
\caption{Individual rates of the \emph{Daphnia} model as suggested by \citet{deRoos1990} and \citet{Kooijman1984}.}
\label{t:functions_Daphnia}
\centering
{\small
\begin{tabular}{llc}
\toprule
{Description} & {Function} \\
\midrule
functional response & $f_0(S) = \frac{\xi S}{1+\xi S}$ \\
growth rate &
$
g(x,S) = \max\{ 0, \gamma_g \left( x_\text{m} f_0(S) - x \right)\}
$ 
\\
mortality rate &
$\mu(x,S) = \mu$
\\
consumption rate &
$\gamma(x,S) = \nu_s f_0(S) x^2$ 
\\
reproduction rate & 
$ 
\beta(x,S) = 
\begin{cases}
0 & \text{if } x_\text{b} \leq x \leq x_\text{A},\\
r_m f_0(S) x^2 & \text{if } x_\text{A}<x<x_\text{max}
\end{cases}
$ 
\\
\bottomrule
\end{tabular}
}
\end{table}

\begin{table}[ht]
\caption{Parameter values of the \emph{Daphnia} model, taken from \citet{deRoos1990} and \citet{Kooijman1984}.}
\label{t:parameters}
\centering
{\small
\begin{tabular}{lcl}
\toprule
{Description} & {Symbol} & {Value} \\
\midrule
length at birth & $x_\text{b}$ & $0.8\ \text{mm}$ \\
length at maturation & $x_\text{A}$ & varying (mm) \\ 
maximum attainable length & $x_\text{max}$ & varying (mm) \\ 
time constant of growth & $\gamma_g$ & $0.15\ \text{d}^{-1}$  \\
shape parameter of functional response & $\xi$ & $7.0\ \text{ml}\cdot\text{cell}^{-1}$ \\
max feeding rate per unit area & $\nu_s$ & $1.8\times\,10^{-6}\,\text{mm}^{-2}\,\text{cell}\,\text{d}^{-1}$ \\
max reproduction rate per unit area & $r_m$ & $0.1\ \text{mm}^{-2}\cdot\text{cell}\cdot\text{d}^{-1}$ \\
max algal growth rate rate & $a_1$ & $0.5\ \text{d}^{-1}$ \\
mortality rate parameter & $\mu$ & $0.2\ \text{d}^{-1}$ \\ 
carrying capacity of the environment & $K$ & varying (cell$\cdot$ml$^{-1}$) \\
\bottomrule
\end{tabular}
}
\end{table}

\begin{figure}[ht]
    \centering
    \input{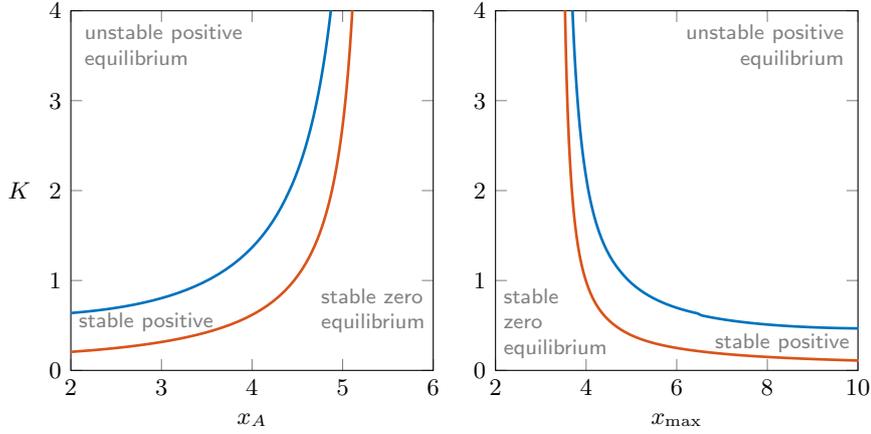}   
    \caption{Left: Stability diagram in the plane $(x_A,K)$, for $x_{\text{max}}=6$; right: Stability diagram in the plane $(x_{\text{max}},K)$, for $x_A=2.5$. }
    \label{fig:daphnia-stability}
\end{figure}

\begin{figure}[ht]
    \centering
    \input{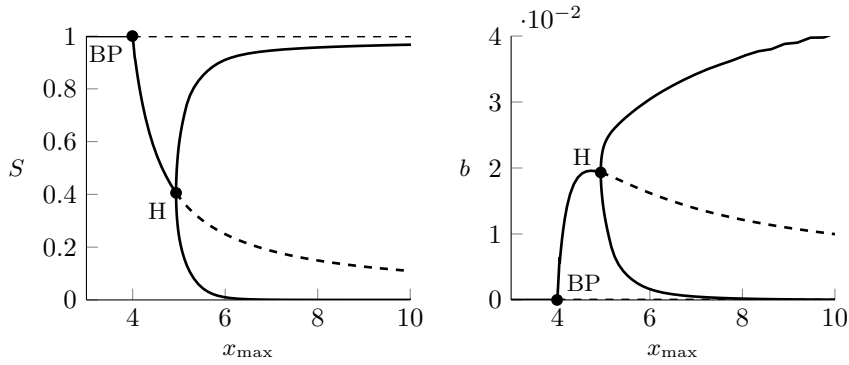}   
    \caption{Bifurcation diagram varying $x_{\text{max}}$, for $x_A=2.5$ and $K=1$, showing the equilibrium branches and the max/min value of the stable periodic orbits emerging from Hopf. }
    \label{fig:daphnia-bifurcation}
\end{figure}


\begin{figure}[ht]
    \centering
    \input{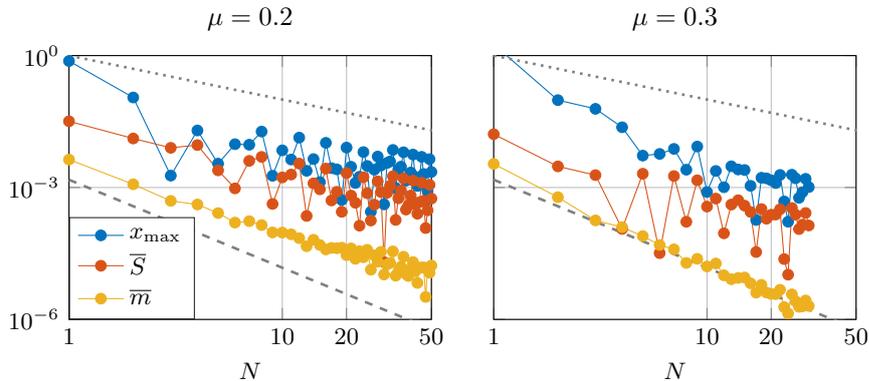}   
    \caption{Log-log plot of the numerical errors in the approximated values of $x_{\text{max}}$, $\overline{S}$ and $\overline{m}$ at the detected Hopf point, for $x_A=2.5$, $K=1$, and $\mu=0.2$ (left) and $\mu=0.3$ (right). The gray lines represent the order of convergence $N^{-1}$ (dotted) and $N^{-2}$ (dashed). }
    \label{fig:daphnia-error}
\end{figure}

Figure~\ref{fig:daphnia-stability} shows the existence and stability boundaries of the nontrivial steady state of \eqref{ODE}--\eqref{tildeS}. The nontrivial steady state destabilizes through a Hopf bifurcation when $K$ or $x_{\text{max}}$ increase, or $x_A$ decreases. 
To illustrate the existence of a stable periodic solution outside of the stability region, we performed a one-parameter continuation with respect to $x_{\text{max}}$, fixing $x_A=2.5$ and $K=1$. 
Figure~\ref{fig:daphnia-bifurcation} shows the Hopf bifurcation on the positive equilibrium branch, from which a branch of stable periodic orbits is born. 

To analyze the convergence of the bifurcations when increasing $N$, we considered the approximation of the Hopf point detected during the continuation in $x_{\text{max}}$ in Figure~\ref{fig:daphnia-bifurcation}. The error, computed with respect to the value obtained with $N=51$, is shown in the left panel of Figure~\ref{fig:daphnia-error}. The experimental order of convergence is $O(N^{-k})$, with $k\in (1,2)$. The convergence order is related to the smoothness of $\overline{m}$, as explained by \cite{Scarabel2021Vietnam}: when $\mu=0.2$, the function $\overline{m}'$ is discontinuous at $x=x_A$. In contrast, for $\mu=0.3$ the equilibrium $\overline{m}$ is continuously differentiable, with second derivative of bounded variation: the experimental order of convergence is approximately two, as shown in the right panel of Figure~\ref{fig:daphnia-error}.

\CCLsection{Concluding Remarks}
\citet{DiekmannScarabelAge} considered the simplest form of $i$-state dynamics. In the present chapter we went one level up and focused on the situation in which the rate of $i$-state change depends on both the $i$-state itself and on a one-dimensional quantity, characterizing the external world. By assuming that all individuals are born with the same $i$-state, we retained some of the simplicity of age-dependent $p$-dynamics. In particular, we found once again that there are two ways of organizing the bookkeeping of $p$-level dynamics, viz., either by keeping track of the $i$-state distribution or by updating the history of the birth rate. But in the latter case one needs, in addition, the full history of the environmental variable in order to determine the current state of individuals born before the point in time at which our model description starts to apply.

It is easy to characterize steady states in terms of the roots of a scalar equation. It is a bit more involved to derive, essentially by formal linearization, a characteristic equation, cf.~\citet{Diekmann2010, Diekmann2017erratum}. But it is downright difficult to prove the principle of linearized stability, in particular to show that the steady state is asymptotically stable if all roots of the characteristic equation have negative real part and unstable if at least one root has positive real part. The key obstruction is formed by the fact that the PDE solution operators are not differentiable, since they involve variable translation of the initial condition. And, in general, the initial condition is not smooth enough to make translation a differentiable operation (so the technical difficulty is identical to the one which makes the study of state-dependent delay equations hard). \citet{Barril2022} manage to prove the principle of linearized stability in the delay equation setting under somewhat restrictive assumptions on the model ingredients. As we described at the end of Section~\ref{sec:dynamicalsystem}, it is an open problem to prove the principle of linearized stability under less restrictive assumptions by exploiting that the lack of differentiability resides in a contribution to the solution operator that decays exponentially in time.

Why are the models considered in this chapter interesting from a biological point of view? Because they capture interaction via variable maturation delay! We refer to \citet{DeRoosPersson2013} for a convincing substantiation of this claim. In addition, we like to mention that the explicit introduction of environmental variables forms the basis for a lego-methodology of formulating complex models by coupling more elementary building blocks.

Even though the `Daphnia' models are relatively simple from a biological point of view, they are quite challenging when it comes to performing a numerical bifurcation analysis in order to investigate how dynamics depends on parameters. The main take home message of this chapter is: in this respect, the recent developments concerning pseudospectral approximation have led to major improvements!

\appendix 

\CCLsection{Pseudospectral Approximation of Equations with Infinite Delay }
\label{sec:infinite}
The aim of this appendix is to tempt readers, and in particular numerical analysts, to give the countless technical difficulties involved in approximating solutions of infinite-delay equations some thought. We hope that our description of recent work will catalyze new developments. 

Motivated by the pseudospectral discretization of RE described by \citet{DiekmannScarabelAge}, one might be tempted to use the delayed system~\eqref{DE-bS}, instead of the PDE~\eqref{m-PDE} with~\eqref{S-m-ODE}, for numerical bifurcation analyses. 

The dependence of the parameters (e.g., the death rate) on environmental conditions, which themselves are influenced by the population, means that it is often impossible to bound the delay \emph{a priori}. Hence, the pseudospectral approximation of RE with finite delay described by \citet{DiekmannScarabelAge} cannot be applied. 
As we will briefly explain in this section, working with the unbounded history interval $\mathbb{R}_-:=(-\infty,0]$ involves both theoretical and numerical complications.

From a theoretical point of view, $L^1$ and $C$ are not appropriate state spaces for the dynamical system when the delay is unbounded. Intuitively, this becomes clear when observing that constant functions do not belong to $L^1(\mathbb{R}_-)$. Furthermore, in order to study the stability via the roots of a characteristic equation, one should ensure that the Laplace transform is well defined to the right of a vertical line in $\mathbb{C}$. 
As shown by \citet{Diekmann2012Blending}, the appropriate spaces to develop the dynamical system theory and derive the principle of linearized stability are exponentially weighted function spaces. For RE, the appropriate Banach space is $L^1_\rho(\mathbb{R}_-, \mathbb{R}^m)$, with $\rho>0$, containing all (equivalence classes of) $\mathbb{R}^m$-valued measurable functions $\varphi$ defined over $\mathbb{R}_-$ such that $\theta \mapsto \ee^{\rho\theta}\varphi(\theta)$ is integrable in $\mathbb{R}_-$, with norm
\begin{equation*}
    \|\varphi\|_{1,\rho} = \int_{-\infty}^0 \ee^{\rho\theta} |\varphi(\theta)| \dd\theta. 
\end{equation*}
For DDE, the appropriate Banach space is $C_{0,\rho}(\mathbb{R}_-, \mathbb{R}^m)$, formed by all $\mathbb{R}^m$-valued functions $\psi$ defined on $\mathbb{R}_-$ such that $\theta\mapsto \ee^{\rho\theta}\psi(\theta)$ is continuous and vanishes at minus infinity, with norm 
\begin{equation*}
    \|\psi\|_{\infty,\rho} = \sup_{\theta \in \mathbb{R_-}} \ee^{\rho\theta} |\psi(\theta)|. 
\end{equation*}
The parameter $\rho$ can be chosen inside an interval $0<\rho<\overline{\rho}$, where the maximum value $\overline{\rho}$ is determined by the equation. The (in)stability of a steady state is then determined by the maximum real part of the roots of the characteristic equation that lie in the right-half plane $\{\Re\lambda>-\rho \} \subset \mathbb{C}$. 

From a numerical point of view, one needs to approximate functions defined on the semi-unbounded interval $\mathbb{R}_-$. 
Several methods are available, including domain truncation, suitable coordinate transformation from unbounded to bounded intervals, or spectral methods based on orthogonal polynomials on the unbounded domain, see for instance the book by \citet[Chapter 7]{Shen2011Book}. We here focus on the pseudospectral approach suggested by \cite{Gyllenberg2018} and \cite{Scarabel2024Infinite}, which exploits the structure of the state spaces by approximating the exponentially weighted functions with exponentially weighted polynomials in~$\mathbb{R}_-$. 

\bigskip
To keep the exposition concise and somewhat complementary to \citet{DiekmannScarabelAge}, we here consider a scalar DDE
\begin{equation}\label{DDE}
    \frac{\dd y(t)}{\dd t} = F(y_t), \quad t\geq 0,
\end{equation}
and only hint at the differences that occur when dealing with RE. Here $F$ is a differentiable function that maps $C_{0,\rho}:=C_{0,\rho}(\mathbb{R}_-, \mathbb{R})$ to $\mathbb{R}$. 
The dynamical system corresponding to~\eqref{DDE} maps an initial function $\psi \in C_{0,\rho}$ to the history $y_t \in C_{0,\rho}$ at time~$t$, defined by $y_t(\theta)=y(t+\theta)$ for $\theta \in \mathbb{R}_-$. The generator of the dynamical system is given by differentiation, captured by the relation 
\begin{equation} \label{PDE-yt}
    \frac{\partial y_t}{\partial t}=\frac{\partial y_t}{\partial\theta},
\end{equation}
plus a perturbation that describes the extension rule at $\theta=0$, given by~\eqref{DDE}. 

Let $w$ be the exponential weight function, i.e., $w(\theta)=\ee^{\rho\theta}$. 
An element~$\psi$ of $C_{0,\rho}$ can grow beyond bound at~$-\infty$, as does its interpolating polynomial~$p$. For this reason, from a numerical point of view it is convenient to work with the weighted counterparts~$w\psi$ and $wp$, which vanish at~$-\infty$, improving numerical stability. 
We therefore need to understand how the shift-and-extend rules for the states~$y_t$ translate into rules for the weighted states~$wy_t$.
First, observe that
\begin{equation} \label{diff-w}
    \frac{\dd}{\dd \theta} (w\psi) = w \,\left( \frac{\dd\psi}{\dd\theta} + \rho \, \psi \right)
    \quad \Rightarrow \quad
    w \frac{\dd \psi}{\dd\theta} = \frac{\dd(w\psi)}{\dd\theta} - \rho\,(w\psi).
\end{equation}
Since $w(0)=1$, the rule for extension in $\theta=0$ is still described by~\eqref{DDE}. 
In other words, we can formally write an abstract differential equation for the pair $(y(t),u(t))$ given by $y(t)=y_t(0)=w(0)y_t(0)$ together with the weighted history $u(t)=wy_t$ as follows 
\begin{align}
    \frac{\dd y(t)}{\dd t} &= F(u(t)/w), \label{ADE-inf-0} \\
    \frac{\dd u(t)}{\dd t} &= \left( \frac{\dd}{\dd\theta} - \rho \right) u(t). \label{ADE-inf}
\end{align}

\bigskip
Given a positive integer $N$, the pseudospectral approximation considers a set of nodes $\{\theta_0,\theta_1,\dots,\theta_N\} \subset \mathbb{R}_-$, such that 
\begin{equation*}
    \theta_N<\theta_{N-1}< \cdots<\theta_1<\theta_0=0,
\end{equation*}
and~$N+1$ corresponding variables $y_0(t)$, \dots, $y_N(t)$ which are used to mimic the dynamics of the weighted history $wy_t$ in each node.
Let $\hat{p}$ be the weighted interpolation polynomial 
\begin{equation}\label{hatp}
    \hat{p}(\theta) = w(\theta) \sum_{j=0}^N \frac{\ell_j(\theta) y_j}{w(\theta_j)}, \quad \theta \in \mathbb{R}_-,
\end{equation}
where the Lagrange polynomials $\ell_j$ are defined by 
\begin{equation*}
    \ell_j(\theta) = \prod_{\substack{i=0\\ i\neq j}}^N \frac{\theta-\theta_i}{\theta_j-\theta_i}. 
\end{equation*}
Note that $\hat{p}$ depends on time via the coefficients $y_j$, but we omit this in the notation. 

By collocation, we assume that $\hat{p}$ satisfies \eqref{ADE-inf-0} at $\theta=\theta_0$, and \eqref{ADE-inf} at $\theta=\theta_1,\dots,\theta_N$. 
To do this, we first note that, for $i=1,\dots,N$, 
\begin{equation*}
    \frac{\dd \hat{p}}{\dd t} \Big|_{\theta=\theta_i} = \frac{\dd}{\dd t} \Big( w(\theta) \sum_{j=0}^N \frac{\ell_j(\theta) y_j}{w(\theta_j)}\Big) \Big|_{\theta=\theta_i} = \sum_{j=0}^N \ell_j(\theta_i) \frac{\dd y_j}{\dd t} = \frac{\dd y_i}{\dd t}
\end{equation*}
and 
\begin{equation*}
    \frac{\dd \hat{p}}{\dd \theta} \Big|_{\theta=\theta_i} = \frac{\dd}{\dd \theta} \Big(w(\theta) \sum_{j=0}^N \frac{\ell_j(\theta) y_j}{w(\theta_j)}\Big ) \Big|_{\theta=\theta_i} = \sum_{j=0}^N \hat{d}_{ij} \,y_j, 
\end{equation*}
where we defined 
\begin{equation*}
    \hat{d}_{ij} := \frac{\dd}{\dd\theta} \left(\frac{w(\theta)\ell_i(\theta)}{w(\theta_i)}\right) \Big|_{\theta=\theta_i}.
\end{equation*}
We conclude that the pseudospectral approximation of \eqref{ADE-inf-0}--\eqref{ADE-inf} is the system of ODE
\begin{equation} \label{ODE-inf}
    \begin{array}{l}
    \displaystyle\frac{\dd y_0}{\dd t} = F(\hat{p}/w), \\[8pt]
    \displaystyle\frac{\dd y_i}{\dd t} = \sum_{j=0}^N \hat{d}_{ij}y_j - \rho \, y_i, \quad i=1,\dots,N,
    \end{array}
\end{equation}
where the first equation captures the rule for extension, and the remaining equations capture translation. The variable~$y_0(t)$ mimics the solution $y(t)$ of~\eqref{DDE}. 

For the RE $y(t)=F(y_t)$ the scheme differs slightly, as one should consider the primitive of the state, $u(t)=w\int_0^\cdot y_t $, similarly as described for finite delay by \citet{DiekmannScarabelAge}. In this case $y_0\equiv 0$, and the variables $y_i$, which approximate $u(t)$ in each node $\theta_i$, satisfy the system of ODE
\begin{equation*} 
    \frac{\dd y_i}{\dd t} = \sum_{j=1}^N \hat{d}_{ij} y_j - \rho \, y_i - w(\theta_i) F\left( \frac{\dd(\hat{p}/w)}{\dd\theta}\right) , \quad i=1,\dots,N,
\end{equation*}
where $\hat{p}$ is still given by \eqref{hatp}. Note that, from a computational point of view, we can exploit \eqref{diff-w} to rewrite 
\begin{equation*}
    \frac{\dd(\hat{p}/w)}{\dd\theta} = \frac{1}{w} \left(\frac{\dd\hat{p}}{\dd\theta} -\rho\,\hat{p} \right).
\end{equation*} 
For details we refer to \citet{Scarabel2024Infinite}, who introduce an abstract framework capturing both DDE and RE. 

\citet[Theorem 4.3]{Scarabel2024Infinite}  show that steady states of~\eqref{DDE} and \eqref{ODE-inf} are in one-to-one correspondence, and linearization at an equilibrium and pseudospectral discretization commute. Hence, to study whether the approximating system effectively captures the stability of steady states, one can focus on linear problems. The convergence of characteristic roots is proved if the collocation nodes are chosen as 
\begin{equation} \label{Laguerre-nodes}
    \theta_j = - \frac{x_j}{2\rho}, \quad j=1,\dots,N,
\end{equation}
where $\{x_j \colon j=1,\dots,N\}$ are either the zeros or extrema of the standard Laguerre polynomials in $[0,+\infty)$, orthogonal with respect to the weight~$\ee^{x/2}$. 
For these families of nodes, the authors prove that for each isolated characteristic root~$\lambda$ of the delayed linear equation such that $\Re\lambda>-\rho$, there exists a sequence $\{\lambda_N\}_N$ of characteristic roots of the approximating ODE that converges to $\lambda$ as $N \to \infty$. The order of convergence is exponential in~$N$, and depends on the multiplicity of the characteristic root, its modulus, and its real part \citep[Theorem~4.5 and~4.6]{Scarabel2024Infinite}. 
We note that, due to limited theoretical results on the error bounds of polynomial interpolation in the norm $\|\cdot\|_{\infty,\rho}$ (see for instance \citet{MastroianniMilovanovicBook}), the proof of convergence follows a different approach compared to the case of finite delay as elaborated in the book by \citet{BMVbook}. 

\bigskip
From a computational point of view, the weighted polynomial interpolant~$\hat{p}$ and the coefficients~$\hat{d}_{ij}$, which are the entries of the \emph{weighted differentiation matrix}, can be computed with efficient and stable algorithms, see for instance \citet{Weideman2000matlab}. 

The function~$F$ usually involves integrals on~$\mathbb{R}_-$. When the kernels are continuous, one can use suitable quadrature rules defined on the collocation nodes, which avoid the numerical evaluation of~$\hat{p}$ outside of the nodes. A typical choice is to use the Gauss--Radau quadrature on the scaled Laguerre extrema~\eqref{Laguerre-nodes}. 

The convergence order of the chosen quadrature rule might limit the accuracy of the approximation of the characteristic roots, hence it is convenient to exploit the flexibility in the choice of the parameter~$\rho$, which affects the collocation nodes via~\eqref{Laguerre-nodes}, to reduce the quadrature error. 
For a linear function $F$ of the form
\begin{equation*}
    F(\psi) = \int_{-\infty}^0 \ee^{\mu\theta} \psi(\theta) \dd\theta, 
\end{equation*}
a convenient choice is $\rho=\mu/2$. In this case, the Gauss--Radau quadrature formula
\begin{equation*}
    \int_{-\infty}^0 f(\theta) \dd\theta \approx \sum_{j=0}^N f(\theta_j) \omega_j, 
\end{equation*}
where $\omega_j$ are the quadrature weights, is exact on functions of the form $f(\theta) = \ee^{2\rho\theta} q(\theta)$, where $q$ is a polynomial of degree $2N$. 
Nonlinear functions $F$ are usually in the form of a finite-dimensional function applied to a linear functional, hence the choice of~$\rho$ can be guided by similar observations. 

A piecewise approach can be more convenient when the kernels are discontinuous, for instance when the life cycle has discrete stages. In this case, using different quadrature rules on the bounded and unbounded intervals can achieve higher accuracy. 


Numerical tests suggest that the approximation of infinite-delay equations requires in general higher-dimensional ODE systems compared to finite delay (and, for size-structured models, to the corresponding PDE) to reach a given accuracy. 
To offer a preliminary insight, we compared a 100-point continuation of the endemic equilibrium using MatCont with the pseudospectral approximation of the infinite-delay system~\eqref{DE-bS} (with the solution of the implicitly-defined maturation age at every step) and the PDE system \eqref{M-PDE} and \eqref{S-m-ODE} (piecewise approach), with $N=10$, so both systems are approximated with an ODE of dimension $2N+1=21$. The difference in computation time is more than $100$-fold, with $124$~seconds required by the infinite-delay approach and $0.7$~seconds by the PDE approach. A more in-depth performance comparison is currently in progress. 

To address the challenge of computational efficiency, \citet{ScarabelVermiglio2025} have recently performed numerical investigations using truncated Laguerre interpolation and quadrature rules, that in many cases allow to reduce the approximating system's dimension while ensuring a given accuracy. In particular, truncated rules seem to be very effective for integration kernels with limited smoothness, while less so for analytic kernels. 

Finally, the presence of an additional parameter to be tuned, $\rho$, that also affects the numerical approximation via the quadrature rules, makes this approximation method less user-friendly in practical applications. 
These challenges are some of the reasons why infinite-delay equations have not yet been integrated in the delay equation importer for MatCont, recently implemented for equations with finite delay by \citet{Liessi-matcont}.

To conclude, the approximation of infinite-delay equations involves several numerical challenges that remain open and, as of today, create barriers to a user-friendly as well as effective implementation and wider adoption of the methods (and consequently of the mathematical models).

\bibliographystyle{plainnat}
\bibliography{size-structure.bib}

\end{document}